\documentclass[amsfonts,amsmath,floatfix,twocolumn]{revtex4}

\usepackage{graphicx}
\usepackage{epsfig}
\usepackage{epstopdf}
\usepackage[usenames,dvipsnames]{xcolor}
\usepackage{braket}
\usepackage{amsthm, mathtools}
\usepackage{amssymb}
\usepackage{mdframed}
\usepackage{bm}
\usepackage{enumitem}
\usepackage{dsfont}
\usepackage{bbold}
\usepackage{eurosym}
\usepackage{microtype} 
\usepackage{hyperref}

\allowdisplaybreaks

\DeclareMathAlphabet{\mymathbb}{U}{BOONDOX-ds}{m}{n}
\DeclareMathAlphabet{\mathbbm}{U}{bbm}{m}{n}

\DeclareMathOperator{\poly}{poly\,}

\DeclareMathOperator{\Tr}{Tr}

\DeclareMathOperator{\Per}{Per}
\DeclareMathOperator{\Haf}{Haf\,}
\DeclareMathOperator{\lHaf}{lHaf\,}
\DeclareMathOperator{\Det}{Det}

\newcommand{\be}{\begin{equation}}
\newcommand{\ee}{\end{equation}}
\newcommand{\ba}{\begin{aligned}}
\newcommand{\ea}{\end{aligned}}

\newcommand{\bc}{\begin{center}}
\newcommand{\ec}{\end{center}}
\newcommand{\beq}{\begin{equation}}
\newcommand{\eeq}{\end{equation}}
\newcommand{\beqq}{\begin{equation*}}
\newcommand{\eeqq}{\end{equation*}}
\newcommand{\beqa}{\begin{align}}
\newcommand{\eeqa}{\end{align}}
\newcommand{\barr}{\begin{array}}
\newcommand{\earr}{\end{array}}
\newcommand{\bi}{\begin{itemize}}
\newcommand{\ei}{\end{itemize}}

\newcommand{\Prb}{\ensuremath{\,\mathrm{Pr}}}

\newtheorem{lem}{Lemma}

\newtheorem{theo}{Theorem}

\newtheorem{coro}{Corollary}
\newtheorem{defi}{Definition}

\begin{document}

\title{Classical simulation of Gaussian quantum circuits with non-Gaussian input states}

\author{Ulysse Chabaud$^{1,2}$}
\email{ulysse.chabaud@gmail.com}
\author{Giulia Ferrini$^3$}
\author{Fr\'ed\'eric Grosshans$^2$}
\author{Damian Markham$^{2,4}$}
\address{$^1$Universit\'e de Paris, IRIF, CNRS, France}
\address{$^2$Sorbonne Universit\'e, CNRS, LIP6, F-75005 Paris, France}
\address{$^3$Department of Microtechnology and Nanoscience (MC2), Chalmers University of Technology, SE-412 96 Gothenburg, Sweden}
\address{$^4$JFLI, CNRS, National Institute of Informatics, University of Tokyo, Tokyo, Japan}


\begin{abstract}

We consider Gaussian quantum circuits supplemented with non-Gaussian input states and derive sufficient conditions for efficient classical strong simulation of these circuits. In particular, we generalise the stellar representation of continuous-variable quantum states to the multimode setting and relate the stellar rank of the input non-Gaussian states, a recently introduced measure of non-Gaussianity, to the cost of evaluating classically the output probability densities of these circuits. Our results have consequences for the strong simulability of a large class of near-term continuous-variable quantum circuits. 
 
\end{abstract}


\maketitle


\section{Introduction}

\noindent Understanding the origin of quantum advantage is of paramount importance, both at the fundamental and technological level.
Continuous-variable (CV) systems are being recognized as a promising alternative to the use of qubits. On the one hand, unprecedented large CV entangled quantum states, of up to one-million elementary systems, can be deterministically generated~\cite{yokoyama2013optical,Yoshikawa2016}. On the other hand, they offer the potential of increased robustness with respect to noise~\cite{ofek2016}.

Wigner function negativity has been shown to be a necessary resource for quantum advantage with CV quantum computing architectures~\cite{mari2012positive,veitch2013}. 
Since Gaussian states and processes have positive Wigner functions, this necessarily corresponds to the use of non-Gaussian resources.
However, establishing under which conditions non-Gaussianity is also sufficient for quantum advantage~\cite{baragiola2019all}, and when instead non-Gaussian circuits are classically efficiently simulable~\cite{garcia2020efficient}, is still an open question.

In what follows, we analyse the computational power of non-Gaussian states and thus focus on the case where Gaussian circuits and measurements are supplemented with non-Gaussian input states as a computational resource. 
We obtain a classical strong simulation algorithm in the case where the non-Gaussian input state has a bounded support over the Fock basis, which runs in time polynomial in the support size and exponential in the total photon number of the input state. 
Note that any normalised state can be approximated arbitrarily well using states of bounded support simply by considering a renormalised truncation of the state. 

This choice of input non-Gaussian states with bounded support is motivated by the recent characterisation of the structure of non-Gaussian quantum states in the single-mode case using the so-called stellar representation~\cite{chabaud2020stellar}. This characterisation establishes an operational hierarchy of non-Gaussian states: the states of finite stellar rank, i.e., in a finite level of the stellar hierarchy, are the states that can be obtained from the vacuum with a given number of photon additions, together with Gaussian unitary operations.  Alternatively, such states may also be obtained from a state with finite support over the Fock basis---a core state---with a Gaussian unitary operation. In this work, we generalise the stellar representation to the multimode case and relate the cost of classically simulating Gaussian circuits with non-Gaussian input states to the stellar rank of these states. Additionnally, we show that the equivalence between photon addition and core state does not hold in the multimode setting, that is, there exist multimode states with bounded support over the Fock basis that cannot be obtained from the vacuum using a finite number of photon additions and Gaussian unitary operations.

The classical simulation results obtained have consequences for more general CV circuits, since non-Gaussian gates and non-Gaussian measurements can be implemented by Gaussian operations together with non-Gaussian ancillary states~\cite{Gottesman2001,ghose2007non,sabapathy2018states}. In particular, we retrieve the fact that Boson Sampling circuits with a logarithmic number of input photons are strongly simulable~\cite{Aaronson2013}, and we show that Gaussian circuits interleaved with a constant number of photon additions or subtractions can be simulated efficiently classically. These results are reminiscent of similar works in discrete-variable quantum architectures, where Clifford circuits supplemented by few magic states have shown to be classically efficiently simulable~\cite{koh2015further,bu2019efficient}. Additionally, our results allow us to compute the output probability distributions of a wide variety of CV circuits and imply their efficient classical strong simulation when the non-Gaussianity of these circuits is small enough, such as Boson Sampling with unbalanced heterodyne detection~\cite{Chakhmakhchyan2017,chabaud2017continuous}, measurement-based CV circuits~\cite{menicucci2006universal,Menicucci2014} or approximate CVIQP circuits~\cite{douce2017continuous,douce2019probabilistic}, among others.

The rest of the paper is structured as follows. In section~\ref{sec:simu}, we recall notions of classical simulation of quantum computations. In section~\ref{sec:GNGinput}, we define classes of Gaussian circuits with non-Gaussian input states and derive explicit expressions for their output probability densities. In section~\ref{sec:StrongsimuNG}, we obtain a classical algorithm for strong simulation of these circuits, with explicit complexity depending on non-Gaussian parameters of the input state---its support size over the Fock basis and its stellar rank. This allows us to give efficient simulability results when these parameters are small enough with respect to the size of the computation, for various interesting subclasses of Gaussian circuits with non-Gaussian input states. We conclude in section~\ref{sec:conclusion}.


\section{Classical simulation of quantum computations}
\label{sec:simu}

\noindent Depending on the approach used for simulating classically the functioning of quantum devices, several notions of simulability are commonly used. Hereafter, we recall the notions of strong and weak simulation.

\subsection{Strong simulation}

\noindent To each quantum computation is associated a probability distribution from which classical outcomes are sampled. In the case of continuous-variable quantum computations with continuous outcomes, the output probability distribution is replaced by an output probability density. This motivates the following (informal) definition~\cite{terhal2002classical,pashayan2020estimation}:

\begin{defi}[Strong simulation]\label{Strongs}
A quantum computation is \textit{strongly simulable} if there exists a classical algorithm which evaluates its output probability distribution (density) or any of its marginals for any outcome in time polynomial in the size of the quantum computation.
\end{defi}

\noindent This notion of simulability is referred to as strong because it asks more from the classical simulation algorithm than from the quantum computation: the quantum computation is merely sampling from a probability distribution (density), while the classical algorithm has to compute efficiently the exact probabilities. Various relaxations of this definition are possible, allowing the classical evaluation to be approximate rather than exact, or to abort with a small probability.

\subsection{Weak simulation}

\noindent A sampling counterpart to the notion of strong simulation is to ask the classical simulation algorithm to mimic the output of the quantum computation~\cite{terhal2002classical,pashayan2020estimation}. Informally:

\begin{defi}[Weak simulation]\label{Weaks}
A quantum computation is \textit{weakly simulable} if there exists a classical algorithm which outputs samples from its output probability distribution (density) in time polynomial in the size of the quantum computation.
\end{defi}

\noindent Akin to strong simulation, various relaxations of this definition are possible, allowing the classical sampling to be approximate rather than exact, or to abort with a small probability. Hereafter we only consider the definition above.

In the case of continuous-variable quantum computations with continuous outcomes, a weaker requirement is to ask the classical simulation not to sample from the output probability density, but rather from a discretised probability distribution obtained from the probability density by performing an efficient binning of the sample space. Indeed, samples from the output probability density yield samples of such a discretised probability distribution with efficient classical post-processing.

As it turns out, weak simulation is indeed weaker than strong simulation, as was shown in Refs.~\cite{terhal2002classical,pashayan2020estimation}: an efficient classical algorithm for strong simulation provides an efficient classical algorithm for weak simulation (assuming one can efficiently sample from efficiently computable univariate probability distributions over a polynomial number of samples).
For quantum computations yielding continuous classical outcomes, the result still holds with a similar proof for binned discretised probability distributions rather than the corresponding probability density, as long as the discretised probabilities can be computed efficiently from the probability density and have support on a polynomial number of bins for each mode.

In what follows, we consider strong simulation of a large class of CV quantum circuits: Gaussian circuits with non-Gaussian input states.


\section{Gaussian circuits with non-Gaussian input states}
\label{sec:GNGinput}

\subsection{The stellar representation}

\noindent The stellar representation of single-mode continuous-variable quantum states has been introduced in~\cite{chabaud2020stellar}. It establishes a hierarchy among single-mode non-Gaussian pure states based on the number of zeros of their Husimi $Q$ function~\cite{husimi1940some}. It shows that any state of stellar rank $N$, that is, whose Husimi $Q$ function has exactly $2N$ zeros (counted with multiplicity), may be written in the following form:
\be
\hat D(\alpha)\hat S(\xi)\ket C,
\label{stellardec}
\ee
where $\hat D(\alpha)$ is a displacement of amplitude $\alpha\in\mathbb C$, $\hat S(\xi)$ is a squeezing of parameter $\xi\in\mathbb C$ and $\ket C$ is a core state of rank $N$, i.e., a state which has bounded support over the Fock basis with highest Fock state $\ket N$. Additionnally, the stellar rank of a pure quantum state corresponds to the minimal number of photon additions that are necessary to engineer the state from the vacuum, together with Gaussian unitary operations. 

Hereafter, we extend a few definitions from~\cite{chabaud2020stellar} to the multimode case, using bold math for multi-index notations (see Appendix~\ref{app:multiindex}).
First, the stellar function, which provides a representation of multimode pure states as multivariate holomorphic functions:

\begin{defi}[Multimode stellar function]
Let $m\in\mathbb N^*$ and let $\ket{\bm\psi}=\sum_{\bm n\ge\bm0}{\psi_{\bm n}\ket{\bm n}}\in\mathcal H^{\otimes m}$ be a normalised pure state over $m$ modes. The \textit{stellar function} of the state $\ket{\bm\psi}$ is defined as
\be
F_{\bm\psi}^\star(\bm z)=e^{\frac12\|\bm z\|^2}\braket{\bm z^*|\bm\psi}=\sum_{\bm n\ge\bm 0}{\frac{\psi_{\bm n}}{\sqrt{\bm n!}}\bm z^{\bm n}},
\label{stellarfuncmulti}
\ee
for all $\bm z\in\mathbb C^m$, where $\ket{\bm z}=e^{-\frac12\|\bm z\|^2}\sum_{\bm n\ge\bm0}{\frac{\bm z^{\bm n}}{\sqrt{\bm n!}}\ket{\bm n}}\in\mathcal H^{\otimes m}$ is the coherent state of complex amplitude $\bm z$.
\end{defi}

\noindent The following definition also extends naturally from the single-mode case:

\begin{defi}[Multimode core state]
\textit{Multimode core states} are defined as the normalised pure quantum states which have a (multivariate) polynomial stellar function.
\end{defi}

\noindent Like in the single-mode case, these are the states with a finite support over the (multimode) Fock basis. For any $m\in\mathbb N^*$, the set of multimode core states over $m$ modes is dense in the set of normalised states for the trace norm (by considering renormalised cutoff states). 
We also introduce the following definitions:

\begin{defi}[Degree of a multimode core state]
The \textit{degree} of a multimode core state is defined as the degree-sum of its stellar function.
\end{defi}

\noindent This definition generalises to the multimode case the notion of stellar rank for core states in the single-mode case.

\begin{defi}[Support of a multimode core state]
The \textit{support} of a multimode core state is the set of Fock basis states which have nonzero overlap with the core state.
\end{defi}

\noindent For example, the $4$-mode core state $\frac1{\sqrt2}(\ket{1230}+\ket{0001})$ has degree $6$, support size $2$, and its stellar function is given by $\frac1{2\sqrt6}z_1z_2^2z_3^3+\frac1{\sqrt2}z_4$, for all $z_1,z_2,z_3,z_4\in\mathbb C^4$. 

Single-mode states of finite stellar rank---which are the states whose stellar function has a finite number of zeros---have two equivalent representations: they are either obtained from an underlying core state by a Gaussian unitary operation, or they can be obtained from the vacuum by photon additions and Gaussian unitary operations. In the multimode setting however, the stellar function has either no zeros or an uncountable infinite number of zeros~\cite{soto1983wigner}. Moreover,  we show in Lemma~\ref{lem:multiGconv} that the two representations are no longer equivalent in the multimode setting: the class of multimode states that are obtained from a multimode core state by a multimode Gaussian unitary operation is striclty larger than the class of multimode states that can be obtained from the vacuum by photon additions and multimode Gaussian unitary operations. Hence, we generalise the notion of finite stellar rank for multimode states as follows:
\begin{defi}[Multimode stellar rank]
Let $\ket{\bm\psi}=\hat G\ket{\bm C}$, where $\hat G$ is a multimode Gaussian unitary and $\ket{\bm C}$ is a multimode core state. The stellar rank of the multimode state $\ket{\bm\psi}$ is defined as the degree of $\ket{\bm C}$.
\end{defi}
\noindent Additionally, for multimode quantum states which do not admit a decomposition of the form $\hat G\ket{\bm C}$, we define their stellar rank to be $+\infty$. 

As in the single-mode case, it is easily seen that the stellar rank is invariant under Gaussian unitary operations (we refer to~\cite{chabaud2020stellar} for the proofs in the single-mode case).  Similarly, the notion of multimode stellar rank induces a multimode stellar hierachy which is robust with respect to the trace norm. In what follows, we only consider multimode states of finite stellar rank, which form a dense subset of the multimode Hilbert space (since these include the set of multimode core states). 

We now turn to the analysis of the computational power of multimode non-Gaussian states. We consider Gaussian unitary circuits with multimode states of finite stellar rank in input. Since these states are of the form $\hat G\ket{\bm C}$, where $\hat G$ is a given multimode Gaussian unitary and $\ket{\bm C}$ is a given multimode core state, this is equivalent to consider Gaussian unitary circuits with input core states.
In what follows, we show that the degree and support size of a multimode core state are sufficient to quantify the hardness of strongly simulating Gaussian circuits with input core states.

\subsection{$G_{\text{core}}$ circuits}

\noindent We define $G_{\text{core}}$ circuits as the family of Gaussian circuits with Gaussian measurements, supplemented by non-Gaussian multimode core states in the input.

Measuring a state with unbalanced heterodyne detection effectively amounts to squeezing the state and then sampling from its $Q$ function~\cite{chabaud2017continuous}, with a squeezing parameter $\xi\in\mathbb C$ depending on the unbalancing of the detection. Setting $\xi=0$ yields balanced heterodyne detection, while sending $|\xi|=r$ to infinity yields homodyne detection. Any Gaussian measurement can thus be implemented by Gaussian unitary operations and heterodyne detection only, since it can be implemented by Gaussian unitary operations and homodyne detection only~\cite{giedke2002characterization,eisert2003introduction}. Without loss of generality, a Gaussian measurement may thus be written as a tensor product of single-mode balanced heterodyne detections preceded by a Gaussian unitary. $G_{\text{core}}$ circuits are then described by two (multidimensional) parameters: a multimode core state $\ket{\bm C}$ in the input and a Gaussian unitary evolution $\hat G$ (Fig.~\ref{fig:GCI}).
\begin{figure}
\begin{center}
\includegraphics[width=0.7\columnwidth]{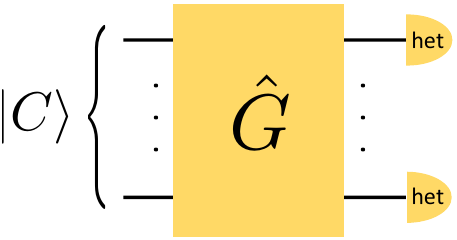}
\caption{Representation of a $G_{\text{core}}$ circuit with multimode core state input $\ket{\bm C}$. The unitary $\hat G$ is Gaussian and the measurement is performed by heterodyne detection.}
\label{fig:GCI}
\end{center}
\end{figure}

In what follows, we derive a general expression for the output probability density functions of these circuits. Then, we study the classical simulability of $G_{\text{core}}$ circuits and of various notable subclasses of these circuits.

\subsection{Output probability density of $G_{\text{core}}$ circuits}

\noindent We first recall a few combinatorial functions related to the permanent, which appear in the expressions of the output probability densities. The hafnian of a square matrix $A=(a_{ij})_{1\le i,j\le2m}$ of size $2m$ is defined as~\cite{caianiello1953quantum}
\begin{equation}
\Haf\,(A):=\sum_{M\in\text{PMP}\,(2m)}{\prod_{\{i,j\}\in M}{a_{ij}}},
\label{Hafnian}
\end{equation}
where the sum is over the perfect matchings of the set $\{1,\dots,2m\}$, i.e., the partitions of $\{1,\dots,2m\}$ in subsets of size $2$. The hafnian of a matrix of odd size is $0$. The hafnian is related to the permanent by
\be
\Haf\begin{pmatrix} \mymathbb0_m & B \\ B^T & \mymathbb0_m \end{pmatrix}=\Per\,(B),
\label{hafper}
\ee
for any $m\times m$ matrix $B$. By convention we set $\Haf\,(\emptyset)=1$, where $\emptyset$ is a square matrix of size $0$.

The loop hafnian of a square matrix $R=(r_{ij})_{1\le i,j\le r}$ of size $r$ is defined as~\cite{bjorklund2019faster}
\be
\lHaf(R):=\sum_{M\in\text{SMP}\,(r)}{\prod_{\{i,j\}\in M}{r_{ij}}},
\label{lHaf}
\ee
where the sum is over the single pair matchings of the set $\{1,\dots,r\}$, defined as the set of perfect matchings of a complete graph with loops with $r$ vertices. This set is isomorphic to the set $\Pi_{1,2}(\{1,\dots,r\})$ of partitions of $\{1,\dots,r\}$ in subsets of size $1$ and $2$ (by mapping a block $\{k\}$ of size $1$ of a partition to the matching $\{k,k\}$ and a block $\{i,j\}$ of size $2$ to the matching $\{i,j\}$). In particular, when $R$ is a matrix whose diagonal entries are all $0$, we have $\lHaf(R)=\Haf\,(R)$. The loop hafnian of a matrix of size $r$ may be computed in time $O(r^32^{r/2})$~\cite{bjorklund2019faster}.

We obtain a closed expression for the output probability density of Gaussian circuits with multimode core states input in Theorem~\ref{th:Pr}, by adapting proof techniques from~\cite{Hamilton2016,kruse2019detailed,quesada2019franck} (see Appendix~\ref{app:multiindex} for multi-index notations):

\begin{theo}\label{th:Pr}
Let $m,n\in\mathbb N$ and let
\be
\ket{\bm C}=\sum_{\substack{\bm p\in\mathbb N^m\\|\bm p|\le n}}{c_{\bm p}\ket{\bm p}},
\ee
be an $m$-mode core state of degree $n$. Let $\hat G$ be a Gaussian unitary over $m$ modes. For all $\bm\alpha\in\mathbb C^m$, let us write $\bm V$ and $\bm{\tilde d}=(\bm d,\bm d^*)$ the covariance matrix and the displacement vector of the Gaussian state $\hat G^\dag\ket{\bm\alpha}$. Then, the output probability density for the $G_{\text{core}}$ circuit $\hat G$ with input $\ket{\bm C}$ and heterodyne detection, evaluated at $\bm\alpha$, is given by
\be
\Prb_{\text{core}}[\bm\alpha]=\kappa(\bm\alpha,\hat G)\!\!\!\!\sum_{\substack{\bm p,\bm q\in\mathbb N^m\\|\bm p|\le n,|\bm q|\le n}}{\!\!\!\!\frac{(-1)^{|\bm p|+|\bm q|}}{\sqrt{\bm p!\bm q!}}c_{\bm p}c_{\bm q}^*\lHaf\left(A_{\bm p,\bm q}\right)},
\label{Pr}
\ee
where
\be
\kappa(\bm\alpha,\hat G)=\frac{\exp\left[-\frac12\bm{\tilde d}^\dag\left(\bm V+\mathbb1_{2m}/2\right)^{-1}\bm{\tilde d}\right]}{\pi^m\sqrt{\Det\,(\bm V+\mathbb1_{2m}/2)}}
\ee
is a Gaussian prefactor and where $A_{\bm p,\bm q}$ is the square matrix of size $|\bm p|+|\bm q|$ obtained from
\be
V=\begin{pmatrix}\mymathbb0_m & \mathbbm1_m \\ \mathbb1_m & \mymathbb0_m \end{pmatrix}\left[\mathbb1_{2m}-\left(\bm V+\mathbb1_{2m}/2\right)^{-1}\right]
\ee
and
\be
D=\left[\bm{\tilde d}^\dag\left(\bm V+\mathbb1_{2m}/2\right)^{-1}\right]^T,
\ee
by replacing the diagonal of $V$ by the elements of $D$, then by repeating $p_k$ times the $k^{th}$ row and column and $q_k$ times the $(m+k)^{th}$ row and column of the obtained matrix for all $k\in\{1,\dots,m\}$.
\end{theo}

\begin{figure}
\begin{center}
\includegraphics[width=0.7\columnwidth]{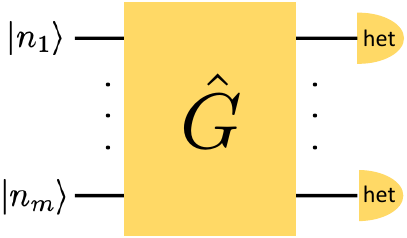}
\caption{Representation of a $G_{\text{Fock}}$ circuit with multimode Fock state input $\ket{n_1\dots n_m}$. The unitary $\hat G$ is Gaussian and the measurement is performed by heterodyne detection.}
\label{fig:GFock}
\end{center}
\end{figure}

\noindent We give a proof in Appendix~\ref{app:thPr}.
When the input core state is a multimode Fock state, we refer to the corresponding subclass of $G_{\text{core}}$ circuits as $G_{\text{Fock}}$ circuits (see Fig.~\ref{fig:GFock}). In that case, the sum in Eq.~(\ref{Pr}) reduces to a single term and we obtain the following result:

\begin{coro}\label{coro:PrFock}
Let $m,n\in\mathbb N$ and let $\bm n=(n_1,\dots,n_m)$ with $|\bm n|=n$. Let $\hat G$ be a Gaussian unitary over $m$ modes. For $\bm\alpha\in\mathbb C^m$, let us write $\bm V$ and $\bm{\tilde d}=(\bm d,\bm d^*)$ the covariance matrix and the displacement vector of the Gaussian state $\hat G^\dag\ket{\bm\alpha}$. Then, the output probability density for the $G_{\text{Fock}}$ circuit $\hat G$ with Fock state input $\ket{\bm n}$ and heterodyne detection, evaluated at $\bm\alpha$, is given by
\be
\Prb_{\text{Fock}}[\bm\alpha]=\frac{\exp\left[-\frac12\bm{\tilde d}^\dag\left(\bm V+\mathbb1_{2m}/2\right)^{-1}\bm{\tilde d}\right]}{\bm n!\pi^m\sqrt{\Det\,(\bm V+\mathbb1_{2m}/2)}}\lHaf(A_{\bm n,\bm n}),
\label{PrFock}
\ee
where $A_{\bm n,\bm n}$ is the square matrix of size $2n$ obtained from
\be
V=\begin{pmatrix}\mymathbb0_m & \mathbbm1_m \\ \mathbb1_m & \mymathbb0_m \end{pmatrix}\left[\mathbb1_{2m}-\left(\bm V+\mathbb1_{2m}/2\right)^{-1}\right]
\ee
and
\be
D=\left[\bm{\tilde d}^\dag\left(\bm V+\mathbb1_{2m}/2\right)^{-1}\right]^T.
\ee
by replacing the diagonal of $V$ by the elements of $D$, then by repeating $n_k$ times the $k^{th}$ and the $(m+k)^{th}$ rows and columns of the obtained matrix for all $k\in\{1,\dots,m\}$.
\end{coro}

\noindent Note that the expressions obtained in Theorem~\ref{th:Pr} and Corollary~\ref{coro:PrFock} may be used to retrieve the expressions of the output probability distributions for a large class of CV circuits that do not necessarily have all their non-Gaussian elements in the input. To see this, consider a $G_{\text{core}}$ circuit of size $2m$ with input $\ket C\otimes\ket{C'}$, where $\ket C$ and $\ket{C'}$ are $m$-mode core states, with Gaussian evolution $\hat G\otimes\mathbb1$, where $\hat G$ is an $m$-mode Gaussian evolution, and projecting onto tensor products of displaced two-mode squeezed states between mode $k$ and $m+k$, for all $k\in\{1,\dots,m\}$. Its output probability density evaluated at $0$, in the limit of infinite squeezing for the two-mode squeezed states, is given by $|\braket{C|\hat G|C'}|^2$, up to a normalisation factor. This encompasses the expressions of output probabilities of Boson Sampling circuits~\cite{Aaronson2013}, when $\hat G$ is a passive linear evolution and $\ket C$ and $\ket{C'}$ are multimode Fock states, or else of Gaussian Boson Sampling circuits~\cite{Hamilton2016}, when $\ket C=\ket0^{\otimes m}$.

\section{Strong simulation of weakly non-Gaussian quantum circuits}
\label{sec:StrongsimuNG}

\noindent In this section, we use the expression obtained in Theorem~\ref{th:Pr} in order to study strong simulation of CV quantum circuits with few non-Gaussian elements. The first general result deals with $G_{\text{core}}$ circuits, i.e., Gaussian circuits with multimode core state input.

\begin{theo}\label{th:strong}
Let $m\in\mathbb N^*$ and let $\ket{\bm C}$ be an $m$-mode core state of support size $s$ and degree $n$. Then, $G_{\text{core}}$ circuits over $m$ modes with input $\ket{\bm C}$ and heterodyne detection can be strongly simulated classically in time $O(s^2n^32^n+\poly m)$.
\end{theo}

\noindent The expression of the probability density $\Prb_{\text{core}}[\bm\alpha]$ in Theorem~\ref{th:Pr} is composed of a Gaussian prefactor multiplied by a sum of $s^2$ loop hafnians of matrices of size at most $2n$, where $s$ is the suppost size and $n$ is the degree of the input core state. The Gaussian prefactor may be computed efficiently in $m$ the number of modes. Thus, to compute the output probability density, one may compute $s^2$ loop hafnians of matrices of size at most $2n$, which can be done in time $O(s^2n^32^n)$~\cite{bjorklund2019faster}. In order to obtain an algorithm for strong simulation, one also needs to compute marginals. We show in Appendix~\ref{app:thstrong} that these may also be computed in time $O(s^2n^32^n+\poly m)$. 

Theorem~\ref{th:strong} implies that strong simulation of $G_{\text{core}}$ circuits is efficient (polynomial in $m$) when the input core state has a logarithmic degree $n=O(\log m)$ and polynomial support size $s=O(\poly m)$. This result may be understood as a generalisation of the efficient classical simulability of Gaussian computations~\cite{bartlett2002efficient}, and has consequences for the simulability of various continuous-variable quantum computing models, in particular those based on Gaussian operations and photon additions or subtractions. We define and consider three interesting examples in what follows: Interleaved Photon-Added Gaussian circuits (IPAG), Interleaved Photon-Subtracted Gaussian circuits (IPSG) and Gaussian circuits with input Fock states ($G_{\text{Fock}}$).

Firstly, we define IPAG circuits with $m$ modes and $n$ photon additions as: (i) product vacuum state over $m$ modes in input, (ii) an evolution composed of interleaved multimode Gaussian unitaries $\hat G^{(0)},\dots,\hat G^{(n)}$ and $n$ single-mode photon additions, and (iii) a Gaussian measurement. Without loss of generality, all the photon additions act on the first mode, since swapping two modes is a Gaussian operation. Moreover, up to an added Gaussian unitary to the final Gaussian unitary $\hat G^{(n)}$, the measurement may be written as a tensor product of balanced heterodyne detections (Fig.~\ref{fig:GaGa}).

\begin{figure}
\begin{center}
\includegraphics[width=\columnwidth]{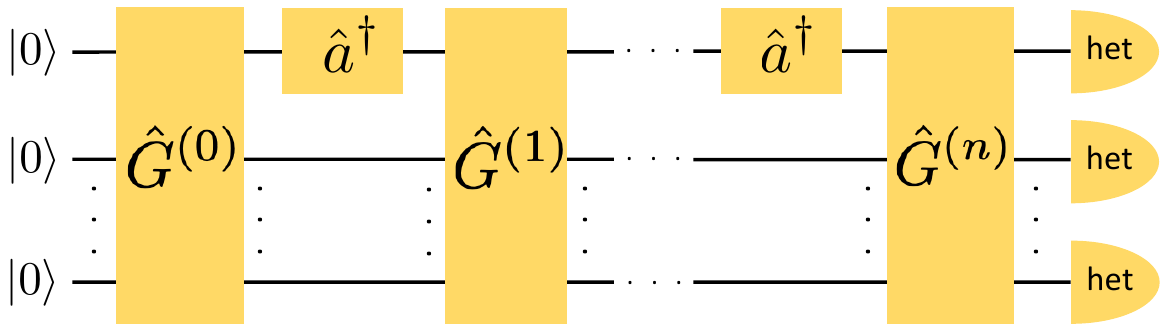}
\caption{Representation of Interleaved Photon-Added Gaussian circuits with $n$ photons additions. The unitaries $\hat G^{(0)},\dots,\hat G^{(n)}$ are Gaussian and the measurement is performed by balanced heterodyne detection. Note that all photon additions act on the first mode without loss of generality, since swapping two modes is a Gaussian operation.}
\label{fig:GaGa}
\end{center}
\end{figure}

The stellar hierarchy of single-mode pure quantum states derived in~\cite{chabaud2020stellar} details the engineering of a single-mode quantum state from the vacuum using unitary Gaussian operations and single photon addition as a non-Gaussian operation. In particular, the states of finite stellar rank, which correspond to the states that can be obtained from the vacuum using a finite number of single photon additions or subtractions, are shown to be exactly the states that are obtained by applying a Gaussian unitary operation to a single-mode core state (see Eq.~(\ref{stellardec})). 

The situation is different in the multimode case: we show that the set of states that can be obtained from a multimode core state with a multimode Gaussian unitary operation is strictly larger than the set of states that can be obtained from the vacuum using a finite number of single photon additions and Gaussian unitary operations. We do so by providing explicitly an example of a state that is a multimode core state, and that is not obtainable with IPAG circuits (Lemma~\ref{lem:multiGconv}). We also deduce strong simulability results for IPAG circuits.

We first establish a reduction to an equivalent model where the evolution and measurement are Gaussian and only the input state is non-Gaussian. This is done by commuting the photon additions to the input of the circuit. The output state of an IPAG circuit with $m$ modes, $n$ photon additions and Gaussian unitaries $\hat G^{(0)},\dots,\hat G^{(n)}$ is given by
\be
\hat G^{(n)}\hat a_1^\dag\hat G^{(n-1)}\hat a_1^\dag\dots\hat G^{(1)}\hat a_1^\dag\hat G^{(0)}\ket0^{\otimes m}.
\label{interleaved1}
\ee
Gaussian operations act on annihilation and creation operators through their symplectic representation, inducing affine transformations of the vector of annihilation and creation operators~\cite{weedbrook2012gaussian}. Let us define the column vector of ladder operators
\be
\bm\lambda=\begin{pmatrix}\hat a_1^\dag\\ \vdots\\ \hat a_m^\dag\\ \hat a_1\\ \vdots\\ \hat a_m\end{pmatrix},
\ee
and let $\hat G$ be an $m$-mode Gaussian operation. Then, there exists a $2m\times2m$ symplectic matrix $\bm S=(s_{ij})_{1\le i,j\le2m}$ and a complex vector $\bm d=(d_1,\dots,d_m)$, such that for all $k\in\{1,\dots,m\}$,
\be
\ba
\hat G\hat a_k^\dag\hat G^\dag&=d_k+(\bm S\bm\lambda)_k\\
&=d_k+\sum_{l=1}^m{s_{k,l}\hat a_l^\dag+s_{k,m+l}\hat a_l},
\ea
\label{commutG}
\ee
where $(\bm S\bm\lambda)_k$ indicates the $k^{th}$ element of the column vector $\bm S\bm\lambda$ and the identity operator is omitted for brevity.
Hence, commuting to the right the creation operators in Eq.~(\ref{interleaved1}), starting by the rightmost one, yields
\be
\ba
\,&\hat G^{(n)}\hat a_1^\dag\dots\hat G^{(1)}\hat a_1^\dag\hat G^{(0)}\ket0^{\otimes m}\\
&=\hat G^{(n)}\hat a_1^\dag\dots a_1^\dag\hat G^{(1)}\hat G^{(0)}\left[d_1^{(0)}+(\bm S^{(0)}\bm\lambda)_1\right]\!\ket0^{\otimes m}\\
&=\dots\\
&=\hat G^{(n)}\dots\hat G^{(0)}\\
&\quad\times\left[d_1^{(n-1)}+(\bm S^{(n-1)}\bm\lambda)_1\right]\dots\left[d_1^{(0)}+(\bm S^{(0)}\bm\lambda)_1\right]\ket0^{\otimes m},
\ea
\label{interleaved2}
\ee
where $\bm S^{(k)}$ and $\bm d^{(k)}=(d^{(k)}_1,\dots,d^{(k)}_m)$ implement the affine transformation corresponding to the action of $(\hat G^{(k)}\hat G^({k-1)}\dots\hat G^{(0)})^\dag$, for all $k\in\{0,\dots,n-1\}$. Writing $\hat G:=\hat G^{(n)}\hat G^{(n-1)}\dots\hat G^{(0)}$ and $\bm S^{(k)}=(s^{(k)}_{i,j})_{1\le i,j\le2m}$ for $k\in\{0,\dots,n-1\}$, we obtain the output state
\be
\hat G\ket{{\bm C}_{\text{IPAG}}},
\label{IPAGoutput}
\ee
where the state
\be
\ba
\ket{{\bm C}_{\text{IPAG}}}&:=\left(d_1^{(n-1)}+\sum_{l=1}^m{s_{1,l}^{(n-1)}\hat a_l^\dag+s_{1,m+l}^{(n-1)}\hat a_l}\right)\\
&\dots\left(d_1^{(0)}+\sum_{l=1}^m{s_{1,l}^{(0)}\hat a_l^\dag+s_{1,m+l}^{(0)}\hat a_l}\right)\ket0^{\otimes m}
\ea
\label{multimodecore}
\ee
is a multimode core state of degree equal to $n$, by property of symplectic matrices. Using this characterisation, we obtain the following result:

\begin{lem}\label{lem:multiGconv}
The set of output states of IPAG circuits is strictly included in the set of output states of $G_{\text{core}}$ circuits.
\end{lem}

\noindent We prove this Lemma in Appendix~\ref{app:multiGconv} by showing that the core state $\frac1{\sqrt2}(\ket{20}+\ket{01})$ cannot be generated by an IPAG circuit. In other words, the set of states that can be obtained from a multimode core state with a multimode Gaussian unitary operation is strictly larger than the set of states that can be obtained from the vacuum using a finite number of single photon additions and Gaussian unitary operations, unlike in the single mode case, where the two sets coincide.

When $n=O(1)$, the support size of the core state $\ket{{\bm C}_{\text{IPAG}}}$ in Eq.~(\ref{multimodecore}) is $O(\poly m)$ and its degree is $O(1)$. Then, from a direct application of Theorem~\ref{th:strong} we obtain:

\begin{lem}\label{IPAG}
IPAG circuits over $m$ modes with $n=O(1)$ photon additions can be strongly simulated efficiently classically.
\end{lem}

\noindent When $n=O(\log m)$ however, the support size of the core state is superpolynomial, so the classical simulation is no longer efficient.  Note that when photon additions are implemented using Gaussian operations and threshold detection, Ref.~\cite{quesada2018gaussian} gives a classical simulation algorithm which has exponential space complexity in the number of photon additions. 

Similarly, we define Interleaved Photon-Subtracted Gaussian circuits (IPSG) by replacing photon additions by subtractions in the definition of IPAG circuits. With a similar proof we obtain the following result:

\begin{coro}\label{coro:IPSG}
IPSG circuits over $m$ modes with $n=O(1)$ photon subtractions can be strongly simulated efficiently classically.
\end{coro}

\noindent Note that the same reasoning holds for Gaussian circuits interleaved with both photon additions and subtractions.

A particular subclass of IPAG circuits, where all the photon additions act at the beginning of the circuit, is the class of $G_{\text{Fock}}$ circuits, i.e., Gaussian circuits with Fock state input. In that case, the input is a multimode core state of support size $1$. With Corollary~\ref{coro:PrFock}, we obtain the following result as an immediate consequence of Theorem~\ref{th:strong}:

\begin{lem}\label{th:strongFock}
Let $m\in\mathbb N^*$ and let $\bm n\in\mathbb N^m$, such that $|\bm n|=O(\log m)$. Then, $G_{\text{Fock}}$ circuits over $m$ modes with Fock state input $\ket{\bm n}$ and heterodyne detection can be strongly simulated efficiently classically.
\end{lem}

\noindent In other words, sampling with Gaussian measurements over $m$ modes from $n=O(\log m)$ indistinguishable photons is strongly simulable classically. This contrasts with the case where $m=\Omega(\poly n)$: in that case, strong simulation and even weak simulation is classically hard~\cite{chabaud2017continuous}. 

Note that, when restricting $G_{\text{Fock}}$ circuits to Boson Sampling circuits~\cite{Aaronson2013} using projection onto two-mode squeezed states as described in the previous section, we retrieve the fact that computing classically the output probabilities is efficient for a logarithmic number of input photons.


\section{Discussion and conclusions}
\label{sec:conclusion}

\noindent In this work, we generalised the notion of stellar rank to the multimode setting and we studied the simulatability of Gaussian circuits with multimode non-Gaussian input states of finite stellar rank, based on the properties of their underlying core states. In particular, we have shown that $G_{\text{core}}$ circuits over $m$ modes, with input states possessing a support of size $n=O(\poly m)$ over the Fock basis and which stellar function has degree $n=O(\log m)$ can be strongly simulated efficiently classically. 

Note that this result, formalised in Theorem~\ref{th:strong}, outperforms existing previous results available in the litterature.
In particular, in Ref.~\cite{pashayan2015estimating} it is shown that the cost for classically estimating the probability of a specific outcome of a quantum circuit---an easier task than sampling or strong simulatability---scales polynomially with the Wigner negativity of the circuit. In terms of the more commonly used Wigner logarithmic negativity~\cite{albarelli2018resource}, that result could be reformulated by saying that the cost for classically estimating the probability of a specific outcome of a quantum circuit scales exponentially with the Wigner logarithmic negativity of the circuit~\footnote{Note that formally the results of Ref.~\cite{pashayan2015estimating} hold for discrete-dimensional systems rather than CV. However, similar results in the CV case could be obtained by performing a discretisation of the phase space~\cite{veitch2013}.
}. For some core states of degree $O(\log m)$, such as the Fock state with $\log m$ photons in $\log m$ modes and vacuum in the other $m-\log m$ modes, the Wigner logarithmic negativity is $\Omega(\log m \log ( \log(m)))$. Therefore, the classical cost for estimating outcome probabilities with the simulation algorithm from Ref.~\cite{pashayan2015estimating} becomes superpolynomial, i.e., is no longer classically efficient. In contrast, the results on the efficient classical simulatability of Theorem~\ref{th:strong} show that we can simulate efficiently classically these circuits. Also note that our results deal with a stronger notion of simulatability than outcome probability estimation. 

Our results are complementary to those in a recently appeared work~\cite{garcia2020efficient}, where non-Gaussian states with unbounded Wigner negativity supplemented to Gaussian circuits are also shown to be classically efficiently simulable. In that work, the simulatability with input unbounded non-Gaussianity, namely with states characterized by infinite stellar rank, is possible due to the fact that the input states are discrete-variable stabiliser states encoded in CV by means of some bosonic encoding, such as  for instance the Gottesman-Kitaev and Preskill one~\cite{Gottesman2001}.

We have also identified various subclasses of $G_{\text{core}}$ circuits to which our simulation algorithm applies. In particular, we have shown that Gaussian circuits interleaved with a constant number of photon additions or subtractions can be efficiently strongly simulated classically. However, the classical algorithm is no longer efficient when the number of photon additions is logarithmic in the number of modes. This contrasts with the---intuitively equivalent---discrete variable case where strong simulation of Clifford circuit supplemented with a logarithmic number of T gates is classically efficient~\cite{Bravyi2016improved,Yifei2019approximate}. It would be interesting to investigate whether this is a fundamental difference between the discrete- and continuous-variable cases, or else if more efficient classical simulation algorithms may be derived in the case of CV circuits based on photon addition and subtraction.

Since any quantum state may be approximated up to arbitrary precision using core states, we expect to obtain approximate simulation results for circuits with specific input non-Gaussian states, such as cat states or GKP states~\cite{Gottesman2001}.  It would also be interesting to investigate further weaker notions of classical simulability for these circuits, such as weak simulation. Finally, it is an open question whether the measure of non-Gaussianity based on the stellar rank can be related to other operational tasks, analogously to~\cite{howard2017}. We leave these questions for future work.


\subsection*{Acknowledgments}

\noindent G.F.\ acknowledges support from the Swedish Research Council (Vetenskapsr\r adet) through the project grant QuACVA, and from the Knut and Alice Wallenberg Foundation through the Wallenberg Center for Quantum Technology (WACQT). D.M.\ and F.G.\ acknowledge funding from the ANR through the ANR-17-CE24-0035 VanQuTe project.


\bibliographystyle{apsrev}
\bibliography{bibliography}


\widetext

\newpage

\appendix


\begin{center}
\LARGE \textbf{Appendix}
\end{center}

In this appendix we provide the notations and proofs for the results stated in the main text.


\section{Multi-index notations}
\label{app:multiindex}

\noindent We use bold math for multimode states, vectors and multi-index notations. Let $m,n\in\mathbb N^*$. We define $\bm0=(0,\dots,0)$ and $\bm1=(1,\dots,1)$, and we write $\bm0^n=(0,\dots,0)\in\mathbb N^n$ or $\bm1^n=(1,\dots,1)\in\mathbb N^n$ to avoid ambiguity. For all $k\in\{1,\dots,m\}$, we also define $\bm1_k=(0,\dots,0,1,0,\dots,0)$, where the $k^{th}$ entry is $1$ and all the other $m-1$ entries are $0$. For all $\bm z=(z_1,\dots,z_m)\in\mathbb C^m$, all $\bm z'=(z_1',\dots,z_m')\in\mathbb C^m$ and all $\bm p=(p_1,\dots,p_m)\in\mathbb N^m$ we write
\be
\ba
\,&\bm z^*=(z_1^*,\dots,z_m^*)\\
\,&-\bm z=(-z_1,\dots,-z_m)\\
\,&\bm{\tilde z}=\bm z\oplus\bm z^*=(z_1,\dots,z_m,z_1^*,\dots,z_m^*)\\
\,&\ket{\bm z}=\ket{z_1\dots z_m}\\
\,&\|\bm z\|^2=|z_1|^2+\cdots+|z_m|^2\\
\,&\bm z^{\bm p}=z_1^{p_1}\dots z_m^{p_m}\\
\,&\bm z+\bm z'=(z_1+z'_1,\dots,z_m+z'_m)\\
\,&\bm z\le\bm z'\Leftrightarrow z_k\le z'_k\quad\forall k\in\{1,\dots,m\}\\
\,&\bm p!=p_1!\dots p_m!\\
\,&|\bm p|=p_1+\cdots+p_m\\
\,&\partial^{\bm p}=\partial_1^{p_1}\dots\partial_m^{p_m}\\
\,&\left(\frac\partial{\partial\bm z}\right)^{\bm p}=\frac{\partial^{|\bm p|}}{\partial z_1^{p_1}\cdots\partial z_m^{p_m}}.
\ea
\label{multiindex}
\ee
%


\section{Proof of Theorem~\ref{th:Pr}}
\label{app:thPr}

\noindent We first prove an intermediate technical result:

\begin{lem}\label{lem:efficientC}
Let $m\in\mathbb N^*$, let $V$ be a $2m\times 2m$ symmetric matrix and let $D$ be a column vector of size $2m$. For all $\bm p,\bm q\in\mathbb N^m$, there exists a square matrix $A_{\bm p,\bm q}(V,D)$ of size $|\bm p|+|\bm q|$ such that
\be
\ba
T_{\bm p,\bm q}(V,D)&:=\int_{\bm\beta\in\mathbb C^m}{\exp\left[\frac12\bm{\tilde\beta}^TV\bm{\tilde\beta}+D^T\bm{\tilde\beta}\right]\left(\frac{\partial}{\partial\bm\beta}\right)^{\bm p}\left(\frac{\partial}{\partial\bm\beta^*}\right)^{\bm q}\delta^{2m}(\bm\beta,\bm\beta^*)\,d^m\!\bm\beta\,d^m\!\bm\beta^*}\\
&\,\,=(-1)^{|\bm p|+|\bm q|}\lHaf\left[A_{\bm p,\bm q}(V,D)\right],
\ea
\ee
assuming the integral is well defined. The matrix $A_{\bm p,\bm q}(V,D)$ is obtained by repeating the entries of $V$ according to $\bm p$ and $\bm q$ and replacing the diagonal of the matrix obtained by the corresponding elements of $D$ (a detailed example follows the proof).
\end{lem}

\begin{proof}

Writing $\bm p=(p_1,\dots,p_m)$ and $\bm q=(q_1,\dots,q_m)$, we first get rid of the integral by successive integration by parts:
\begin{align}
\nonumber T_{\bm p,\bm q}(V,D)&=(-1)^{|\bm p|+|\bm q|}\left(\frac{\partial}{\partial\bm\beta}\right)^{\bm p}\left(\frac{\partial}{\partial\bm\beta^*}\right)^{\bm q}\exp\left[\frac12\bm{\tilde\beta}^TV\bm{\tilde\beta}+D^T\bm{\tilde\beta}\right]\Bigg\rvert_{\bm{\tilde\beta}=\bm0}\\
&=(-1)^{|\bm p|+|\bm q|}\prod_{j=1}^m{\left(\frac{\partial}{\partial\beta_j}\right)^{p_j}\left(\frac{\partial}{\partial\beta_j^*}\right)^{q_j}}\exp\left[\frac12\bm{\tilde\beta}^TV\bm{\tilde\beta}+D^T\bm{\tilde\beta}\right]\Bigg\rvert_{\bm{\tilde\beta}=\bm0}\\
\nonumber&=(-1)^{|\bm p|+|\bm q|}\prod_{j\in\mathcal E_{\bm p,\bm q}}{\left(\frac{\partial}{\partial\tilde\beta_j}\right)}\exp\left[\frac12\bm{\tilde\beta}^TV\bm{\tilde\beta}+D^T\bm{\tilde\beta}\right]\Bigg\rvert_{\bm{\tilde\beta}=\bm0},
\end{align}
where the multiset $\mathcal E_{\bm p,\bm q}$ is defined as the set of size $|\bm p|+|\bm q|$ obtained from $\{1,\dots,2m\}$ by repeating $p_k$ times the index $k$ and $q_k$ times the index $m+k$, for all $k\in\{1,\dots,m\}$.

We make use of Fa\`a di Bruno's formula~\cite{hardy2006combinatorics} in order to expand the product of partial derivatives and we obtain
\be
T_{\bm p,\bm q}(V,D)=(-1)^{|\bm p|+|\bm q|}\sum_{\pi\in\Pi(\mathcal E_{\bm p,\bm q})}\prod_{B\in\pi}{\left(\frac{\partial^{|B|}}{\prod_{j\in B}\partial\tilde\beta_j}\right)}\left[\frac12\bm{\tilde\beta}^TV\bm{\tilde\beta}+D^T\bm{\tilde\beta}\right]\Bigg\rvert_{\bm{\tilde\beta}=\bm0},
\label{Faadi2}
\ee
where $\Pi(\mathcal E_{\bm p,\bm q})$ denotes the set of all partitions of the multiset $\mathcal E_{\bm p,\bm q}$, and where the product runs over the blocks $B$ of the partition $\pi\in\Pi(\mathcal E_{\bm p,\bm q})$, with $|B|$ the size of the block. The function $\bm{\tilde\beta}^\dag V\bm{\tilde\beta}+D^\dag\bm{\tilde\beta}$ is a sum of a quadratic and a linear functions, so all derivatives of order greater than $2$ in the sum vanish. We thus have
\be
\ba
T_{\bm p,\bm q}(V,D)&=(-1)^{|\bm p|+|\bm q|}\sum_{\pi\in\Pi_{1,2}(\mathcal E_{\bm p,\bm q})}\prod_{B\in\pi}{\left(\frac{\partial^{|B|}}{\prod_{j\in B}\partial\tilde\beta_j}\right)}\left[\frac12\bm{\tilde\beta}^TV\bm{\tilde\beta}+D^T\bm{\tilde\beta}\right]\Bigg\rvert_{\bm{\tilde\beta}=\bm0}\\
&=(-1)^{|\bm p|+|\bm q|}\sum_{\pi\in\Pi_{1,2}(\mathcal E_{\bm p,\bm q})}\prod_{\{i,j\}\in\pi}{\left(\frac{\partial^2}{\partial\tilde\beta_i\partial\tilde\beta_j}\right)}\left[\frac12\bm{\tilde\beta}^TV\bm{\tilde\beta}+D^T\bm{\tilde\beta}\right]\Bigg\rvert_{\bm{\tilde\beta}=\bm0}\\
&\quad\quad\quad\quad\quad\quad\quad\quad\quad\quad\quad\quad\times\prod_{\{k\}\in\pi}{\left(\frac{\partial}{\partial\tilde\beta_k}\right)}\left[\frac12\bm{\tilde\beta}^TV\bm{\tilde\beta}+D^T\bm{\tilde\beta}\right]\Bigg\rvert_{\bm{\tilde\beta}=\bm0},
\ea
\ee
where $\Pi_{1,2}(\mathcal E_{\bm p,\bm q})$ denotes the set of all partitions of the multiset $\mathcal E_{\bm p,\bm q}$ in subsets of size $1$ and $2$. All derivatives of order $2$ of the linear term vanish, and all derivatives of order $1$ of the quadratic term vanish when evaluated at $\bm{\tilde\beta}=\bm0$. We thus obtain
\be
T_{\bm p,\bm q}(V,D)=(-1)^{|\bm p|+|\bm q|}\sum_{\pi\in\Pi_{1,2}(\mathcal E_{\bm p,\bm q})}\prod_{\{i,j\}\in\pi}{\left(\frac{\partial^2}{\partial\tilde\beta_i\partial\tilde\beta_j}\right)}\left[\frac12\bm{\tilde\beta}^TV\bm{\tilde\beta}\right]\Bigg\rvert_{\bm{\tilde\beta}=\bm0}\prod_{\{k\}\in\pi}{\left(\frac{\partial}{\partial\tilde\beta_k}\right)}\left[D^T\bm{\tilde\beta}\right]\Bigg\rvert_{\bm{\tilde\beta}=\bm0}.
\ee
Writing $V=(v_{ij})_{1\le i,j\le2m}$, with $V=V^T$, and $D=(d_k)_{1\le k\le2m}$ we obtain
\be
T_{\bm p,\bm q}(V,D)=(-1)^{|\bm p|+|\bm q|}\sum_{\pi\in\Pi_{1,2}(\mathcal E_{\bm p,\bm q})}{\prod_{\{i,j\}\in\pi}{v_{ij}}\prod_{\{k\}\in\pi}{d_k}}.
\label{TmnVDsum}
\ee
We now show that this expression may be rewritten as the loop hafnian of a matrix of size $|\bm p|+|\bm q|$. 
Define $V_{\bm p,\bm q}$ the $(|\bm p|+|\bm q|)\times(|\bm p|+|\bm q|)$ matrix obtained from $V$ by repeating $p_k$ times its $k^{th}$ rows and columns and $q_k$ times its $(m+k)^{th}$ rows and columns, for $k\in\{1,\dots,m\}$. Similarly, define $D_{\bm p,\bm q}$ the column vector of size $|\bm p|+|\bm q|$ obtained from $D$ by repeating $p_k$ times its $k^{th}$ element and $q_k$ times its $(m+k)^{th}$ element, for $k\in\{1,\dots,m\}$. Finally, let $A_{\bm p,\bm q}(V,D)=(a_{ij})_{1\le i,j\le|\bm p|+|\bm q|}$ be the $(|\bm p|+|\bm q|)\times(|\bm p|+|\bm q|)$ matrix obtained from $V_{\bm p,\bm q}$ by replacing its diagonal with the vector $D_{\bm p,\bm q}$. Then, Eq.~(\ref{TmnVDsum}) rewrites
\be
\ba
T_{\bm p,\bm q}(V,D)&=(-1)^{|\bm p|+|\bm q|}\sum_{\pi\in\Pi_{1,2}(\{1,\dots,|\bm p|+|\bm q|\})}{\prod_{\{i,j\}\in\pi}{a_{ij}}\prod_{\{k\}\in\pi}{a_{kk}}}\\
&=(-1)^{|\bm p|+|\bm q|}\sum_{M\in\text{SMP}(|\bm p|+|\bm q|)}{\prod_{\{i,j\}\in M}{a_{ij}}}\\
&=(-1)^{|\bm p|+|\bm q|}\lHaf\left[A_{\bm p,\bm q}(V,D)\right],
\ea
\label{lHafA}
\ee
where the sum in the first line is over the partitions of $\{1,\dots,|\bm p|+|\bm q|\}$ in subsets of size $1$ and $2$, where the sum in the second line is over the single pair matchings of the set $\{1,\dots,|\bm p|+|\bm q|\}$ and where the third line comes from the definition of the loop hafnian in Eq.~(\ref{lHaf}). 

\end{proof}

\noindent Let us illustrate with an example how the matrix $A_{\bm p,\bm q}(V,D)$ appearing in Lemma~\ref{lem:efficientC} is constructed from the matrix $V$ and the vector $D$. Let us set $m=2$, $\bm p=(2,0)$ and $\bm q=(1,0)$. We write
\be
V=\begin{pmatrix}
v_{11}&v_{12}&v_{13}&v_{14}\\
v_{21}&v_{22}&v_{23}&v_{24}\\
v_{31}&v_{32}&v_{33}&v_{34}\\
v_{41}&v_{42}&v_{43}&v_{44}
\end{pmatrix}\quad\text{and}\quad D=\begin{pmatrix}
d_1\\
d_2\\
d_3\\
d_4
\end{pmatrix}.
\ee
We first build the matrix $V_{\bm p,\bm q}$ by repeating $p_k$ times the $k^{th}$ row and column of $V$ and $q_k$ times the $(m+k)^{th}$ row and column. In that case, $\bm p=(p_1,p_2)=(2,0)$, so we repeat $2$ times the first row and column and discard the second row and column, and $\bm q=(q_1,q_2)=(1,0)$, so we keep the third row and column and discard the fourth row and column, obtaining the $3\times3$ matrix
\be
V_{\bm p,\bm q}=\begin{pmatrix}
v_{11}&v_{11}&v_{13}\\
v_{11}&v_{11}&v_{13}\\
v_{31}&v_{31}&v_{33}
\end{pmatrix}.
\ee
Similarly, we obtain the vector $D_{\bm p,\bm q}$ by repeating $p_k$ times the $k^{th}$ element of $D$ and $q_k$ times the $(m+k)^{th}$ element, as
\be
D_{\bm p,\bm q}=\begin{pmatrix}
d_1\\
d_1\\
d_3
\end{pmatrix}.
\ee
Finally, we replace the diagonal of $V_{\bm p,\bm q}$ by $D_{\bm p,\bm q}$:
\be
A_{\bm p,\bm q}(V,D)=\begin{pmatrix}
d_1&v_{11}&v_{13}\\
v_{11}&d_1&v_{13}\\
v_{31}&v_{31}&d_3
\end{pmatrix}.
\ee
In this construction by repeating rows and columns, the first index denotes which rows and columns are repeated for indices in $\{1,\dots,m\}$, while the second index denotes which rows and columns are repeated for indices in $\{m+1,\dots,2m\}$.

Combining Lemma~\ref{lem:efficientC} with phase space formalism and properties of Gaussian states, we are now ready to prove Theorem~\ref{th:Pr}:

\begin{proof}

The Gaussian circuit is composed of a Gaussian unitary $\hat G$ and balanced heterodyne detection. The output probability density reads, for all $\bm\alpha=(\alpha_1,\dots,\alpha_m)\in\mathbb C^m$,
\be
\ba
\text{Pr}_{\text{core}}[\bm\alpha]&=\Tr\left[\hat G\ket{\bm C}\!\bra{\bm C}\hat G^\dag\Pi_{\bm\alpha}\right]\\
&=\frac1{\pi^m}\Tr\left[\hat G^\dag\ket{\bm\alpha}\!\bra{\bm\alpha}\hat G\ket{\bm C}\!\bra{\bm C}\right]\\
&=\int_{\bm\beta\in\mathbb C^m}{Q_{\hat G^\dag\ket{\bm\alpha}\!\bra{\bm\alpha}\hat G}(\bm\beta)P_{\ket{\bm C}\!\bra{\bm C}}(\bm\beta)\,d^m\!\bm\beta\,d^m\!\bm\beta^*},
\ea
\label{PrIPAG1}
\ee
where $\Pi_{\bm\alpha}=\frac1{\pi^m}\ket{\bm\alpha}\!\bra{\bm\alpha}$ is the POVM element corresponding to the heterodyne detection of $\bm\alpha=(\alpha_1,\dots,\alpha_m)$. The state $\hat G^\dag\ket{\bm\alpha}$ is a Gaussian state: let $\bm V$ be its covariance matrix and $\bm d$ its displacement vector. For all $\bm\gamma\in\mathbb C^m$, we write $\bm{\tilde\gamma}=(\gamma_1,\dots,\gamma_m,\gamma_1^*,\dots,\gamma_m^*)$. Then, for all $\bm\beta\in\mathbb C^m$,
\be
\ba
Q_{\hat G^\dag\ket{\bm\alpha}\!\bra{\bm\alpha}\hat G}(\bm\beta)&=\frac1{\pi^m\sqrt{\Det\,(\bm V+\mathbb1_{2m}/2)}}\exp\left[-\frac12(\bm{\tilde\beta}-\bm{\tilde d})^\dag\left(\bm V+\mathbb1_{2m}/2\right)^{-1}(\bm{\tilde\beta}-\bm{\tilde d})\right]\\
&=\frac{\exp\left[-\frac12\bm{\tilde d}^\dag\left(\bm V+\mathbb1_{2m}/2\right)^{-1}\bm{\tilde d}\right]}{\pi^m\sqrt{\Det\,(\bm V+\mathbb1_{2m}/2)}}\exp\left[-\frac12\bm{\tilde\beta}^\dag\left(\bm V+\mathbb1_{2m}/2\right)^{-1}\bm{\tilde\beta}+\bm{\tilde d}^\dag\left(\bm V+\mathbb1_{2m}/2\right)^{-1}\bm{\tilde\beta}\right],
\ea
\label{QGaussian}
\ee
i.e., it is a Gaussian function which can be computed efficiently. On the other hand, we have
\be
\ket{\bm C}\!\bra{\bm C}=\sum_{\substack{\bm p,\bm q\in\mathbb N^m\\|\bm p|\le n,|\bm q|\le n}}{c_{\bm p}c_{\bm q}^*\ket{\bm p}\!\bra{\bm q}},
\ee
so that
\be
P_{\ket{\bm C}\!\bra{\bm C}}(\bm\beta)=\sum_{\substack{\bm p,\bm q\in\mathbb N^m\\|\bm p|\le n,|\bm q|\le n}}{c_{\bm p}c_{\bm q}^*P_{\ket{\bm p}\!\bra{\bm q}}(\bm\beta)},
\label{Pcore}
\ee
for all $\bm\beta\in\mathbb C^m$. Moreover we have, for all $\bm p,\bm q\in\mathbb N^m$ and all $\bm\beta\in\mathbb C^m$,
\be
\ba
P_{\ket{\bm p}\!\bra{\bm q}}(\bm\beta)&=\frac{e^{\|\bm\beta\|^2}}{\sqrt{\bm p!\bm q!}}\left(\frac{\partial}{\partial\bm\beta}\right)^{\bm p}\left(\frac{\partial}{\partial\bm\beta^*}\right)^{\bm q}\delta^{2m}(\bm\beta,\bm\beta^*)\\
&=\frac{e^{\frac12\bm{\tilde\beta}^\dag\bm{\tilde\beta}}}{\sqrt{\bm p!\bm q!}}\left(\frac{\partial}{\partial\bm\beta}\right)^{\bm p}\left(\frac{\partial}{\partial\bm\beta^*}\right)^{\bm q}\delta^{2m}(\bm\beta,\bm\beta^*),
\ea
\label{Ppq}
\ee
where $\delta^{2m}(\bm\beta,\bm\beta^*)=\delta(\beta_1)\cdots\delta(\beta_m)\,\delta(\beta_1^*)\cdots\delta(\beta_m^*)$. Combining Eqs.~(\ref{QGaussian}), (\ref{Pcore}) and (\ref{Ppq}) with Eq.~(\ref{PrIPAG1}) we obtain
\be
\ba
\text{Pr}_{\text{core}}[\bm\alpha]&=\kappa(\bm\alpha,\hat G)\sum_{\substack{\bm p,\bm q\in\mathbb N^m\\|\bm p|\le n,|\bm q|\le n}}\frac{c_{\bm p}c_{\bm q}^*}{\sqrt{\bm p!\bm q!}}\int_{\bm\beta\in\mathbb C^m}\Bigg\{\exp\left[-\frac12\bm{\tilde\beta}^\dag\left(\bm V+\mathbb1_{2m}/2\right)^{-1}\bm{\tilde\beta}\right]\\
&\times\exp\left[\bm{\tilde d}^\dag\left(\bm V+\mathbb1_{2m}/2\right)^{-1}\bm{\tilde\beta}\right]e^{\frac12\bm{\tilde\beta}^\dag\bm{\tilde\beta}}\left(\frac{\partial}{\partial\bm\beta}\right)^{\bm p}\left(\frac{\partial}{\partial\bm\beta^*}\right)^{\bm q}\delta^{2m}(\bm\beta,\bm\beta^*)\Bigg\}\,d^m\!\bm\beta\,d^m\!\bm\beta^*,
\ea
\label{PrIPAG2}
\ee
where we have set
\be
\kappa(\bm\alpha,\hat G)=\frac{\exp\left[-\frac12\bm{\tilde d}^\dag\left(\bm V+\mathbb1_{2m}/2\right)^{-1}\bm{\tilde d}\right]}{\pi^m\sqrt{\Det\,(\bm V+\mathbb1_{2m}/2)}}.
\ee
Given that
\be
\bm{\tilde\beta}^\dag=\bm{\tilde\beta}^T\begin{pmatrix} \mymathbb0_m & \mathbbm1_m \\ \mathbb1_m & \mymathbb0_m\end{pmatrix},
\ee
for all $\bm\beta\in\mathbb C^m$, the integral terms in Eq.~(\ref{PrIPAG2}) rewrite as
\be
\int_{\bm\beta\in\mathbb C^m}{\exp\left[\frac12\bm{\tilde\beta}^TV\bm{\tilde\beta}+D^T\bm{\tilde\beta}\right]\left(\frac{\partial}{\partial\bm\beta}\right)^{\bm p}\left(\frac{\partial}{\partial\bm\beta^*}\right)^{\bm q}\delta^{2m}(\bm\beta,\bm\beta^*)\,d^m\!\bm\beta\,d^m\!\bm\beta^*},
\label{PrIPAG3}
\ee
for $|\bm p|\le n$ and $|\bm q|\le n$, where
\be
V=\begin{pmatrix}\mymathbb0_m & \mathbbm1_m\\ \mathbb1_m & \mymathbb0_m \end{pmatrix}\left[\mathbb1_{2m}-\left(\bm V+\mathbb1_{2m}/2\right)^{-1}\right]
\label{V}
\ee
is a $2m\times 2m$ symmetric matrix, due to the initial structure of the covariance matrix, and where
\be
D=\left[\bm{\tilde d}^\dag\left(\bm V+\mathbb1_{2m}/2\right)^{-1}\right]^T
\label{D}
\ee
is a column vector of size $2m$. By Lemma~\ref{lem:efficientC}, the terms in Eq.~(\ref{PrIPAG3}) are equal to
\be
(-1)^{|\bm p|+|\bm q|}\lHaf\left(A_{\bm p,\bm q}\right),
\ee
where the square matrices $A_{\bm p,\bm q}$ of size $|\bm p|+|\bm q|$ are obtained from $V$ by repeating its entries according to $\bm p$ and $\bm q$ and replacing the diagonal by the corresponding elements of $D$ (see the example following Lemma~\ref{lem:efficientC} for a detailed description of the construction). With Eq.~(\ref{PrIPAG2}) we finally obtain
\be
\text{Pr}_{\text{core}}[\bm\alpha]=\kappa(\bm\alpha,\hat G)\sum_{\substack{\bm p,\bm q\in\mathbb N^m\\|\bm p|\le n,|\bm q|\le n}}{\frac{(-1)^{|\bm p|+|\bm q|}}{\sqrt{\bm p!\bm q!}}c_{\bm p}c_{\bm q}^*\lHaf\left(A_{\bm p,\bm q}\right)},
\ee
where
\be
\kappa(\bm\alpha,\hat G)=\frac{\exp\left[-\frac12\bm{\tilde d}^\dag\left(\bm V+\mathbb1_{2m}/2\right)^{-1}\bm{\tilde d}\right]}{\pi^m\sqrt{\Det\,(\bm V+\mathbb1_{2m}/2)}},
\ee
where $\bm V$ and $\bm d$ are the covariance matrix and the diplacement vector of the Gaussian state $\hat G^\dag\ket{\bm\alpha}$, respectively.

\end{proof}


\section{Proof of Theorem~\ref{th:strong}}
\label{app:thstrong}

\begin{proof}

By Theorem~\ref{th:Pr}, up to an efficiently computable prefactor, the output probability density is a sum of $s^2$ loop hafnians, where $s$ is the support size of the input core state. The loop hafnian of a matrix of size $r$ may be computed in time $O(r^32^{r/2})$~\cite{bjorklund2019faster}. For $|\bm p|\le n$ and $|\bm q|\le n$, the matrices $A_{\bm p,\bm q}$ appearing in Eq.~(\ref{Pr}) are efficiently computable square matrices of size $|\bm p|+|\bm q|\le2n$, so all the loop hafnians may be computed in time $O(n^32^n)$. Hence, the output probability density can be evaluated in time $O(s^2n^32^n+\poly m)$.

We now consider the evaluations of the marginal probability densities. Let $k\in\{1,\dots,m-1\}$, for all $\bm\alpha=(\alpha_1,\dots,\alpha_k)\in\mathbb C^k$ we have
\begin{align}
\nonumber \text{Pr}_{\text{core}}[\bm\alpha]&=\Tr\left[\hat G\ket{\bm C}\!\bra{\bm C}\hat G^\dag\left(\Pi_{\bm\alpha}\otimes\mathbb 1_{m-k}\right)\right]\\ 
&=\frac1{\pi^k}\Tr\left[\hat G^\dag\left(\ket{\bm\alpha}\!\bra{\bm\alpha}\otimes\mathbb 1_{m-k}\right)\hat G\ket{\bm C}\!\bra{\bm C}\right]\\
\nonumber &=\pi^{m-k}\int_{\bm\beta\in\mathbb C^m}{Q_{\hat G^\dag\left(\ket{\bm\alpha}\!\bra{\bm\alpha}\otimes\mathbb 1_{m-k}\right)\hat G}(\bm\beta)\,P_{\ket{\bm C}\!\bra{\bm C}}(\bm\beta)\,d^m\!\bm\beta\,d^m\!\bm\beta^*},
\end{align}
where $\Pi_{\bm\alpha}=\frac1{\pi^k}\ket{\alpha_1,\dots,\alpha_k}\!\bra{\alpha_1,\dots,\alpha_k}$ is the POVM element corresponding to the heterodyne detection of $(\alpha_1,\dots,\alpha_k)$ over the first $k$ modes. With Lemma~\ref{lem:efficientC} and the proof of Theorem~\ref{th:Pr}, it is sufficient to show that $Q_{\hat G^\dag\left(\ket{\bm\alpha}\!\bra{\bm\alpha}\otimes\mathbb 1_{m-k}\right)\hat G}$ is an efficiently computable Gaussian function in order to prove that the marginal probability density can be evaluated in time $O(s^2n^32^n+\poly m)$. 

For all $(\alpha_1,\dots,\alpha_k)\in\mathbb C^k$ and all $(\gamma_1,\dots,\gamma_{m-k})\in\mathbb C^{m-k}$ we write $\bm\alpha=(\alpha_1,\dots,\alpha_k,0,\dots,0)\in\mathbb C^m$ and $\bm\gamma=(0,\dots,0,\gamma_1,\dots,\gamma_{m-k})\in\mathbb C^m$ so that $\bm\alpha+\bm\gamma=(\alpha_1,\dots,\alpha_k,\gamma_1,\dots,\gamma_{m-k})\in\mathbb C^m$. Using the overcompleteness of coherent states we obtain, for all $(\alpha_1,\dots,\alpha_k)\in\mathbb C^k$ and for all $\bm\beta\in\mathbb C^m$,
\be
\pi^{m-k}Q_{\hat G^\dag\left(\ket{\bm\alpha}\!\bra{\bm\alpha}\otimes\mathbb 1_{m-k}\right)\hat G}(\bm\beta)=\int_{\bm\gamma=(\gamma_1,\dots,\gamma_{m-k})\in\mathbb C^{m-k}}{Q_{\hat G^\dag\ket{\bm\alpha+\bm\gamma}\!\bra{\bm\alpha+\bm\gamma}\hat G}(\bm\beta)\,d^{m-k}\bm\gamma d^{m-k}\bm\gamma^*}.
\label{integralQ}
\ee
Let $S$ and $\bm{\tilde d}=(\bm d,\bm d^*)$ be the symplectic matrix and the displacement vector associated with the Gaussian unitary $\hat G^\dag$. The Gaussian state
\be
\hat G^\dag\ket{\alpha_1,\dots,\alpha_k,\gamma_1,\dots,\gamma_{m-k}}=\hat G^\dag\ket{\bm\alpha+\bm\gamma}
\ee
is described by the covariance matrix $\bm V=\frac12SS^\dag$ and the displacement vector $S(\bm{\tilde\alpha}+\bm{\tilde\gamma})+\bm{\tilde d}$. Its $Q$ function is thus given by
\be
Q_{\hat G^\dag\ket{\bm\alpha+\bm\gamma}\!\bra{\bm\alpha+\bm\gamma}\hat G}(\bm\beta)=\frac{\exp\left[-\frac12(\bm{\tilde\beta}-S(\bm{\tilde\alpha}+\bm{\tilde\gamma})-\bm{\tilde d})^\dag\left(\bm V+\mathbb1_{2m}/2\right)^{-1}(\bm{\tilde\beta}-S(\bm{\tilde\alpha}+\bm{\tilde\gamma})-\bm{\tilde d})\right]}{\pi^m\sqrt{\Det\,(\bm V+\mathbb1_{2m}/2)}},
\ee
for all $(\alpha_1,\dots,\alpha_k)\in\mathbb C^k$, for all $(\gamma_1,\dots,\gamma_{m-k})\in\mathbb C^{m-k}$ and for all $\bm\beta\in\mathbb C^m$. Let us discard the efficiently computable denominator and expand the product in the exponential. Writing $M=\left(\bm V+\mathbb1_{2m}/2\right)^{-1}$, we are left with
\be
\exp\left[-\frac12(\bm{\tilde\beta}-S\bm{\tilde\alpha}-\bm{\tilde d})^\dag M(\bm{\tilde\beta}-S\bm{\tilde\alpha}-\bm{\tilde d})\right]\cdot\exp\left[-\frac12\bm{\tilde\gamma}^\dag S^\dag MS\bm{\tilde\gamma}+(\bm{\tilde\beta}-S\bm{\tilde\alpha}-\bm{\tilde d})^\dag MS\bm{\tilde\gamma}\right],
\ee
The first exponential term is an efficiently computable Gaussian function which factors out of the integral in Eq.~(\ref{integralQ}). Rewriting Eq.~(\ref{integralQ}) up to this efficiently computable Gaussian function we are left with
\be
\ba
\int_{\bm\gamma=(0,\dots,0,\gamma_1,\dots,\gamma_{m-k})\in\mathbb C^m}&{\exp\left[-\frac12\bm{\tilde\gamma}^\dag S^\dag MS\bm{\tilde\gamma}+(\bm{\tilde\beta}-S\bm{\tilde\alpha}-\bm{\tilde d})^\dag MS\bm{\tilde\gamma}\right]d^{m-k}\bm\gamma d^{m-k}\bm\gamma^*}\\
&=\int_{\bm\gamma=(\gamma_1,\dots,\gamma_{m-k})\in\mathbb C^{m-k}}{\exp\left[-\frac12\bm{\tilde\gamma}^TV\bm{\tilde\gamma}+D^T\bm{\tilde\gamma}\right]d^{2(m-k)}\bm{\tilde\gamma}},
\ea
\ee
where $V$ is the $2(m-k)\times2(m-k)$ submatrix of
\be
\begin{pmatrix} \mymathbb0_m & \mathbbm1_m \\ \mathbb1_m & \mymathbb0_m \end{pmatrix}S^\dag MS
\ee
obtained by removing the rows and colums of indices $l$ and $m+l$ for $l\in\{1,\dots,k\}$, and where $D$ is the column vector of size $2(m-k)$ obtained by removing the elements of
\be
\left[(\bm{\tilde\beta}-S\bm{\tilde\alpha}-\bm{\tilde d})^\dag MS\right]^T
\ee
of indices $l$ and $m+l$ for $l\in\{1,\dots,k\}$. The matrix $V$ and the vector $D$ are efficiently computable. Moreover,
\be
\int_{\bm\gamma=(\gamma_1,\dots,\gamma_{m-k})\in\mathbb C^{m-k}}{\exp\left[-\frac12\bm{\tilde\gamma}^TV\bm{\tilde\gamma}+D^T\bm{\tilde\gamma}\right]d^{2(m-k)}\bm{\tilde\gamma}}=\frac{(2\pi)^{m-k}}{\sqrt{\Det\,(V)}}\exp\left[\frac12D^TV^{-1}D\right],
\ee
which is an efficiently computable Gaussian function of $\bm\beta$.

This implies that the value of the marginal probability density $\Pr\,[\alpha_1,\dots,\alpha_k]$ may be computed in time $O(s^2n^32^n+\poly m)$. Moreover, it is clear that this does not depent on the choice of $k\in\{1,\dots,m-1\}$ and on the choice of the modes. Hence, all marginal probability densities may be evaluated in time $O(s^2n^32^n+\poly m)$.

\end{proof}


\section{Proof of Lemma~\ref{lem:multiGconv}}
\label{app:multiGconv}

\begin{proof}

The inclusion is immediate with Eq.~(\ref{IPAGoutput}). Up to the Gaussian unitary, it is sufficient to consider core states. To prove the strict inclusion, we show that the $m$-mode core state $(\ket{20}+\ket{01})\otimes\ket0^{\otimes m-2}$ (we omit normalisation), which has degree $2$, is not a core state of the form of Eq.~(\ref{multimodecore}).

By Eq.~(\ref{multimodecore}), all $m$-mode core states of IPAG circuits of degree $2$ have the form
\be
\left(d^{(1)}+\sum_{k=1}^m{s_k^{(1)}\hat a_k^\dag+s_{m+k}^{(1)}\hat a_k}\right)\left(d^{(0)}+\sum_{l=1}^m{s_l^{(0)}\hat a_l^\dag+s_{1,m+l}^{(0)}\hat a_l}\right)\ket0^{\otimes m},
\ee
for some complex numbers $d^{(0)},d^{(1)},s_1^{(0)},\dots,s_{2m}^{(0)},s_1^{(1)},\dots,s_{2m}^{(1)}$. This expression rewrites 
\be
\left(d^{(1)}+\sum_{k=1}^m{s_k^{(1)}\hat a_k^\dag+s_{m+k}^{(1)}\hat a_k}\right)\left(\sum_{l=1}^m{s_l^{(0)}\ket{\bm1_l}}+d^{(0)}\ket{\bm0}\right),
\ee
where for all $l\in\{1,\dots,m\}$, we write $\bm 1_l=(0,\dots,0,1,0\dots,0)$, with a $1$ at the $l^{th}$ position. We finally obtain
\be
\sqrt2\sum_{k=1}^m{s_k^{(0)}s_k^{(1)}\ket{\bm2_k}}+\sum_{\substack{k,l=1\\k\neq l}}^m{s_k^{(0)}s_l^{(1)}\ket{\bm1_k+\bm1_l}}+\sum_{k=1}^m{\left(d^{(1)}s_k^{(0)}+d^{(0)}s_k^{(1)}\right)\ket{\bm1_k}}+\left(d^{(0)}d^{(1)}+\sum_{k=1}^m{s_k^{(0)}s_{m+k}^{(1)}}\right)\ket{\bm0},
\label{IPAGcore}
\ee
where for all $k\in\{1,\dots,m\}$, we write $\bm 2_k=(0,\dots,0,2,0\dots,0)$, with a $2$ at the $k^{th}$ position.
On the other hand we have
\be
(\ket{20}+\ket{01})\otimes\ket0^{\otimes m-2}=\ket{\bm2_1}+\ket{\bm1_2}.
\label{weirdcore}
\ee
In order for this core state to be of the form of Eq.~(\ref{IPAGcore}) we must have
\be
\begin{cases}
s_1^{(0)}s_1^{(1)}\neq0\\
s_k^{(0)}s_l^{(1)}=0,\text{ for }k\neq l,
\end{cases}
\ee
by considering the first and second terms of Eq.~(\ref{IPAGcore}). This implies $s_k^{(0)}=s_k^{(1)}=0$ for all $k\neq1$. Hence, the coefficient of $\ket{\bm1_2}$ in Eq.~(\ref{IPAGcore}) is equal to $0$, while it is nonzero in Eq.~(\ref{weirdcore}). Therefore the core state described by Eq.~(\ref{weirdcore}) cannot be generated by an IPAG circuit.

\end{proof}


\end{document}